\documentclass{article}

\usepackage{times}
\usepackage{algorithmic}
\usepackage{amsmath}
\usepackage{amssymb}

\newcommand{\MEE}[1]{\ensuremath{\prblemname{MEE}}\left(#1\right)}
\newcommand{\lMEE}[2]{\ensuremath{\prblemname{MEE}^{#1}_{#2}}}
\newcommand{\lSIZE}[1]{\ensuremath{\mathop{size}_{#1}}}
\newcommand{\pair}[2]{\ensuremath{\langle#1,#2\rangle}}

\newcommand{\GF}[1]{\ensuremath{\mathtext{GF}\left(#1\right)}}

\newcommand{\clone}[1]{\ensuremath{\left[ #1\right]}}
\newcommand{\clonename}[1]{\mathrm{#1}}

\newcommand{\cV}{\ensuremath{\clonename{V}}}
\newcommand{\cE}{\ensuremath{\clonename{E}}}
\newcommand{\cL}{\ensuremath{\clonename{L}}}

\newcommand{\dual}[1]{\mathrm{dual}\!\left(#1\right)}

\newcommand{\complexityclassname}[1]{\ensuremath{\mathrm{#1}}}
\newcommand{\complementclass}[1]{\complexityclassname{co}#1}
\newcommand{\SigPH}[1]{\Sigma^p_{#1}}
\newcommand{\CONP}{\complementclass{\NP}}
\newcommand{\NP}{\complexityclassname{NP}}
\newcommand{\PTIME}{\complexityclassname{P}}
\newcommand{\redlogm}{\ensuremath{\leq_{m}^{\log}}}
\newcommand{\redpm}{\ensuremath{\leq_{m}^{p}}}
\newcommand{\eqlogm}{\ensuremath{\equiv_{m}^{\log}}}
\newcommand{\var}[1]{\ensuremath{\mathrm{VAR}\!\left(#1\right)}}
\newcommand{\set}[1]{\ensuremath\left\{#1\right\}}
\newcommand{\prblemname}[1]{\ensuremath{\mathsf{#1}}}
\newcommand{\mathtext}[1]{\ensuremath{\mathrm{\text{#1}}}}
\newcommand{\sat}[1]{\prblemname{SAT}\!\left(#1\right)}
\newcommand{\nmodels}{\ensuremath{\not\models}}
\newcommand{\card}[1]{\left| #1 \right|}
\newcommand{\decisionproblem}[3]{
\medskip
\vspace*{1mm}
\begin{tabular}{ll}
\textit{Problem:} & #1 \\
\textit{Input:} & #2 \\
\textit{Question:} & #3
\end{tabular}
\smallskip
\vspace*{1mm}
}

\newcounter{theorem}
\newtheorem{theorem}{Theorem}[section]
\newtheorem{corollary}[theorem]{Corollary}
\newtheorem{example}[theorem]{Example}
\newtheorem{lemma}[theorem]{Lemma}
\newtheorem{proposition}[theorem]{Proposition}
\newenvironment{definition}{\paragraph{Definition}}{\ \\}

\newcounter{prooffactcounter}
\newtheorem{prooffact}[prooffactcounter]{Fact}

\title{Minimization for Generalized Boolean Formulas
\thanks{Supported by the NFS, grants CCR-0311021 and IIS-0713061, the DAAD postdoc program, and by a Friedrich Wilhelm Bessel Research Award. Work done in part while
the second author worked at the Rochester Institute of Technology. This is the full version of~\cite{Hemaspaandra-Schnoor-Minimization-IJCAI-2011-TOAPPEAR}.}
}

\author{Edith Hemaspaandra and Henning Schnoor}
\date{}

\author{Edith Hemaspaandra \\
  \multicolumn{1}{p{.7\textwidth}}{\centering\emph{Rochester Institute of Technology, \\ Rochester, NY, USA}} \ \\ \ \\ Henning Schnoor \\
  \multicolumn{1}{p{.7\textwidth}}{\centering\emph{Christian-Albrechts-Universit\"at zu Kiel, \\ Kiel, Germany}}}

\begin{document}
\maketitle
\sloppy

\begin{abstract}
The minimization problem for propositional formulas is an important optimization problem in the second level of the polynomial hierarchy. In general, the problem is $\SigPH 2$-complete under Turing reductions, but restricted versions are tractable. We study the complexity of minimization for formulas in two established frameworks for restricted propositional logic: The Post framework allowing arbitrarily nested formulas over a set of Boolean connectors, and the constraint setting, allowing generalizations of CNF formulas. In the Post case, we obtain a dichotomy result: Minimization is solvable in polynomial time or \CONP-hard. This result also applies to Boolean circuits. For CNF formulas, we obtain new minimization algorithms for a large class of formulas, and give strong evidence that we have covered all polynomial-time cases.
\end{abstract}

\section{Introduction}

The minimization problem for propositional formulas is one of the most natural optimization problems in the polynomial hierarchy. In fact, a variant of this problem was a major motivation for the definition of the polynomial hierarchy \cite{mest72}. The goal of minimization is to find a minimum equivalent formula to a given input formula. In this paper, we study the \emph{minimum equivalent expression (MEE)} problem, where the input is a formula $\varphi$ and a number $k$, and the question is to determine whether there exists a formula which is equivalent to $\varphi$ and of size at most $k$ (we study different notions of ``size'').

The problem is trivially in $\SigPH 2$, but a better lower bound than \CONP-hardness had been open for many years. In \cite{hewe02}, Hemaspaandra and Wechsung proved the problem to be (many-one) hard for parallel access to \NP. Recently, it was shown to be $\SigPH 2$-complete under Turing reductions by Buchfuhrer and Umans~\cite{bu-um-11:minimization}.

Minimization in restricted fragments of propositional logic has been studied for the case of Horn formulas in order to find small representations of knowledge bases~\cite{Hammer-Kogan-Horn-Minimization}. Prime implicates, a central tool for minimizing Boolean formulas~\cite{Quine-52}, have been used in several areas of artificial intelligence research. We mention~\cite{DBLP:journals/jair/AdjimanCGRS06}, where prime implicates were used in peer-to-peer data management systems for the semantic web, and~\cite{DBLP:journals/logcom/Bittencourt08}, which applies them in the context of belief change operators. Two-level logic minimization is an important problem in logic synthesis~\cite{umans06}. Different variants of minimization have been studied: The problem is $\SigPH 2$-complete for CNF formulas~\cite{uma01}, \NP-complete for Horn formulas~\cite{BC94}, and solvable in \PTIME\ for 2CNF formulas~\cite{chang}.

In this paper we study the complexity of minimization for syntactically restricted formulas. Two frameworks for restricting the expressive power of propositional logic have been used for complexity classifications in recent years:

\begin{itemize}
 \item The \emph{Post framework}~\cite{pos41} considers formulas that instead of the usual operators $\wedge$, $\vee$, and $\neg$, use an arbitrary set $B$ of Boolean functions as connectors. Depending on $B$, the resulting formulas may express only a subset of all Boolean functions, or may be able to express all functions more succinctly than the usual set $\set{\wedge,\vee,\neg}$.
 \item The \emph{constraint framework}~\cite{sch78} studies formulas in CNF form, where the types of allowed clauses (e.g., Horn, 3CNF, or XOR clauses) are defined in a \emph{constraint language} $\Gamma$ containing ``templates'' of generalized CNF-clauses that are allowed in so-called $\Gamma$-formulas.
\end{itemize}

In both frameworks, a wide range of complexity classifications has been obtained. For the Post framework, we mention the complexity of satisfiability~\cite{lew79}, equivalence~\cite{rei01}, modal satisfiability~\cite{HemaspaandraSchnoorSchnoor-GeneralizedModal-JCSS-2010}, and non-monotonic logics~\cite{DBLP:journals/corr/abs-1009-1990}. In the constraint setting, besides the satisfiability problem~\cite{sch78,abi+09}, also enumeration of solutions~\cite{crhe97}, equivalence and isomorphism~\cite{boherevo02,boherevo04}, circumscription~\cite{nordh-johnsson:2004}, and unique satisfiability~\cite{ju99} have been studied, see~\cite{crvo08} for a survey. The complexity of satisfiability for non-Boolean domains is also a very active field, see e.g.,~\cite{DBLP:journals/jacm/Bulatov06,DBLP:conf/dagstuhl/BulatovV08}.

For many considered problems, ``dichotomy results'' were achieved, proving that every choice of $B$ or $\Gamma$ leads to one of the same two complexity degrees, usually polynomial-time solvable and \NP-complete. This is surprising since there are infinitely many sets $B$ and $\Gamma$, and we know that there are, for example, infinitely many degrees of complexity between \PTIME\ and \NP\ cases unless $\PTIME=\NP$~\cite{lad75b}.

A ``Galois Connection'' between constraint languages and closure properties in the Post setting determines the complexity of many computational problems~\cite{jecogy97,SchnoorSchnoor-Galois-LNCS-2008}. In contrast, we show that these tools do not apply to minimization.

In the Post setting, we obtain a complete classification of the tractable cases of the minimization problem: For a set $B$ of Boolean functions, the problem to minimize $B$-formulas is solvable in polynomial time or \CONP-hard, hence avoiding the degrees between \PTIME\ and \CONP-completeness. Our results in this framework apply to both the formula and the circuit case, and to different notions of size of formulas and circuits. 

In the constraint case, we define \emph{irreducible} constraint languages, among which we identify a large class whose formulas can be minimized in polynomial time, and prove \NP- or \CONP-hardness results for most of the remaining cases. More precisely, we prove the following: For an irreducible constraint language for which equivalence can be tested in polynomial time, the minimization problem is \NP-complete if the language can express (dual) positive Horn, and can be solved in polynomial time otherwise. \NP-completeness for the positive Horn case was shown in~\cite{BC94}. Our analysis thus shows that previous hardness results about the hardness of minimizing positive Horn formulas were ``optimal:'' As soon as a CNF fragment of propositional logic is strictly less expressive than positive Horn, formulas can be minimized efficiently. Since irreducibility is a natural condition for constraint languages that are used in knowledge representation, a consequence of our result is that knowledge bases that do not need the full expressive power of positive Horn admit efficient ``compression algorithms.''

Our contribution is threefold: 
\begin{enumerate}
 \item We give new and non-trivial minimization algorithms for large classes of formulas.
 \item In the Post setting, we prove that all remaining cases are \CONP-hard. In the constraint setting, we give strong evidence that larger classes do not have efficient minimization algorithms. 
 \item We show that minimization behaves very differently than many other problems in the context of propositional formulas: The usually-applied algebraic tools for the constraint setting cannot be applied to minimization. Also, complexities in the Post- and constraint framework differ strongly: In particular, the constraint framework contains \NP-complete cases; such cases do not exist in the Post framework (unless $\NP=\CONP$).
\end{enumerate}

\section{Minimization in the Post Framework}

We fix a finite set $B$ of Boolean functions of finite arity. We define \emph{$B$-formulas} inductively: A variable $x$ is a $B$-formula, and if $\varphi_1,\dots,\varphi_n$ are $B$-formulas, and $f$ is an $n$-ary function from $B$, then $f(\varphi_1,\dots,\varphi_n)$ is a $B$-formula. We often identify the function $f$ and the symbol representing it. $\var\varphi$ denotes the set of variables in a formula $\varphi$. We write $\varphi(x_1,\dots,x_n)$ to indicate that $\var\varphi=\set{x_1,\dots,x_n}$. For an assignment $\alpha\colon\var\varphi\rightarrow\set{0,1}$, the \emph{value of $\varphi$ for $\alpha$}, $\varphi(\alpha)$, is defined in the straightforward way. We write $\alpha\models\varphi$ if $\varphi(\alpha)=1$, and say that $\alpha$ \emph{satisfies} $\varphi$. Formulas $\varphi_1$ and $\varphi_2$ are \emph{equivalent} if $\varphi_1(\alpha)=\varphi_2(\alpha)$ for all $\alpha$, we then write $\varphi_1\equiv\varphi_2$. The satisfiability problem for $B$-formulas, i.e., the problem to decide whether a given $B$-formula has at least one solution, is denoted with $\sat{B}$.

Formulas can be succinctly represented as \emph{circuits}, which are essentially DAGs where formulas are trees. Although every circuit can be rewritten into a formula, the size of the resulting formula can be exponential in the size of the circuit.

In the Post framework, we study two variations of the minimization problem that differ in the notion of the \emph{size} of a formula $\varphi$. An obvious way to measure size is the number of occurrences of literals, which we denote with $\lSIZE l(\varphi)$. The second measurement is motivated by the study of Boolean circuits, where the size of a circuit is usually the number of non-input gates. For a formula, this is the number of appearing function symbols. We denote this number with $\lSIZE s(\varphi)$. Our results also hold for obvious variations of these measures (e.g., counting variables instead of occurrences, also counting input gates, etc). For a set $B$ as above, we define:

\decisionproblem{$\lMEE{F/C}{l/s}(B)$}{A $B$-formula/circuit $\phi$ and a natural number $k$}{Is there a $B$-formula/circuit $\psi$ with $\lSIZE{l/s}(\psi)\leq k$ and $\phi\equiv\psi$?}

For an $n$-ary Boolean function $f$, the function $\dual f$ is defined as $\dual f(x_1,\dots,x_n)=\overline{f(\overline{x_1},\dots,\overline{x_n})}$, i.e., $\dual f$ is the function obtained from $f$ by exchanging the roles of the values $0$ and $1$ in the evaluation of $f$. Since the minimization problem is trivially invariant under this transformation, we obtain the following result (as usual, for a set $B$ of Boolean functions, with $\dual B$ we denote the set $\set{\dual f\ \vert\ f\in B}$):

\begin{proposition}\label{prop:lewis duality for mee}
 Let $B$ be a finite set of Boolean functions, then $\lMEE{F/C}{l/s}(B)\eqlogm\lMEE{F/C}{l/s}(\dual B)$ 
\end{proposition}

\subsection{Tractable Cases: Polynomial-Time algorithms}\label{section:lewis:ptime}

An $n$-ary Boolean function $f$ is an OR-function if it is constant or if $f(x_1,\dots,x_n)$ is equivalent to $x_{r_1}\vee x_{r_2}\vee\dots\vee x_{r_m}$ for a subset $\set{x_{r_1},x_{r_2},\dots,x_{r_m}}\subseteq\set{x_1,\dots,x_n}$. AND- and XOR-functions are defined analogously. We show that formulas using only these functions can be minimized easily:

\begin{theorem}\label{theorem:lewis:ptime cases}
 $\lMEE{F/C}{l/s}(B)$ can be solved in polynomial time if $B$ contains only OR-functions, only AND-functions, or only XOR-functions.
\end{theorem}

We mention that the theorem, as all of our results in this section, applies to all four combinations of $F/C$ and $s/l$. We also stress that all algorithms in this paper do not only determine whether a formula with the given size restriction exists, but also compute a minimum equivalent formula.

\begin{proof}
 Let $\star$ denote the binary OR-operator if $B\subseteq\cV$, or the binary XOR-operator if $B\subseteq\cL$ (the case $\cE$ follows from Proposition~\ref{prop:lewis duality for mee}, since $\dual{\cV}=\cE$). We know that every element of $B$ is of the form $c_0\star c_1 x_1\star\dots c_n x_n$, where the $c_i$ indicate which of the $x_i$ is a relevant argument. Note that if $B\subseteq\cV$ and $c_0=1$, then none of the arguments are relevant. Without loss of generality, assume that the first $l$ of the variables are relevant. We then represent the function $f$ with the tuple $(c,l,n)$. We now show how building formulas from the functions in $B$ can be expressed with arithmetic operations on these tuples. Given two formulas representing the functions $f_1=(c_1,l_1,n_1)$ and $f_2=(c_2,k_2,n_2)$, we can, using the operations allowed in superposition, obtain formulas representing the following: 

\begin{description}
 \item[Substituting $f_2$ for a relevant argument of $f_1$]{\ \\$(c_1,l_1,n_1)\circ_{\mathtext{rel}}(c_2,l_2,n_2)=(c_1\star c_2,l_1+l_2-1,n_1+n_2-1)$\\ \emph{applicable iff $l_1\ge 1$}}
 \item[Substituting $f_2$ for an irrelevant argument of $f_1$]{\ \\$(c_1,l_1,n_1)\circ_{\mathtext{rel}}(c_2,l_2,n_2)=(c_1,l_1,n_1+n_2-1)$\\ \emph{applicable iff $l_1<n_1$}}
 \item[Identifying two relevant variables in $f_1$ in the case $B\subseteq\cV$]{\ \\ $(c_1,l_1,n_1)\rightarrow (c_1,l_1-1,n_1)$ \\ \emph{applicable iff $l_1\ge 2$}}
 \item[Identifying two relevant variables in $f_1$ in the case $B\subseteq\cL$]{\ \\ $(c_1,l_1,n_1)\rightarrow (c_1,l_1-2,n_1)$ \\ \emph{applicable iff $l_1\ge 2$}}
\end{description}

Note that identifying a relevant and an irrelevant variable comes down to simply renaming the irrelevant variable, and therefore is not of interest to us.

We now describe the polynomial time algorithm. Assume that we are given a $B$-formula $\varphi(x_1,\dots,x_n)$ and a natural number $k$. For the classes of functions that we are looking at, it is easy to determine which variables of the formula are relevant: In both cases, the $i$-th variable of $\varphi$ is relevant if and only if $$\varphi(\underbrace{0,\dots,0}_{i-1},0,\underbrace{0,\dots,0}_{n-i-1})\neq\varphi(\underbrace{0,\dots,0}_{i-1},1,\underbrace{0,\dots,0}_{n-i-1}).$$ Without loss of generality, assume that the relevant variables of $\varphi$ are exactly the variables $x_1,\dots,x_n$. Note that since $\varphi$ describes a function from $B$, we can again represent $\varphi$ as $(c_{\varphi},l_{\varphi},n_{\varphi})$ as above. The question if we can find a $B$-formula equivalent to $\varphi$ with less then $k$ variable occurrences is the same as the question if we can obtain, from the tuples representing the functions in $B$, a tuple of the form $(c_{\varphi},l_{\varphi},n')$, where $n'\ge n$ and $l_{\varphi}+n'\leq k$ (we can then, by renaming the $n'$ irrelevant variables to $x_{l+1},\dots,x_n$, construct the equivalent formula).

It is obvious that if we have $0$-ary constant functions in $B$, then we can remove irrelevant variable occurrences from any $B$-formula by replacing them with the constant functions, and hence $\varphi,k$ is a positive instance if and only if $l+n\leq k$. Therefore assume without loss of generality that all of the functions in $B$ are of the form $(c,l,n)$ for $n\ge 1$. In this case, the operations defined above are all non-decreasing in $n$. Hence we can simply generate a table containing all entries $(c,l,n)$ for $c\in\set{0,1}$ and $l,n\leq\max\left({\set{n'\ \vert\ \exists c',l' (c',l',n')\in B}\cup\set{n_{\varphi}}}\right)$, where an entry is set to true if and only if a corresponding formula can be built from $B$. We start by setting all entries to true which correspond to functions in $B$, and then simply apply the operations defined above until no changes occur anymore in the table, then we check if an entry as required is set to true. Since $n$ is smaller than the input, the table is of polynomial size, and the procedure obviously can be performed in polynomial time.

Note that this also gives a polynomial time procedure if the set $B$ is part of the input, if the functions in $B$ are given using the formula representations (essentially, the tuples $(c,l,n)$ in unary).
\end{proof}

\subsection{Hardness Results: Relationship to Satisfiability}

The satisfiability problem for the formulas covered in Section~\ref{section:lewis:ptime} can easily be solved in polynomial time. We now show that this is indeed a prerequisite for a tractable minimization problem---formally, we prove that the complement of the satisfiability problem (i.e., the set of all binary strings that are not positive instances of $\sat{B}$) reduces to the minimization problem.

\begin{theorem}\label{theorem:unsat red mee lewis}
 For every finite set $B$ of Boolean functions, $\overline{\sat{B}}\redlogm\lMEE{F/C}{l/s}(B)$.
\end{theorem}

\begin{proof}
 Without loss of generality, assume that there is an unsatisfiable $B$-formula $\psi$ (if such a formula does not exist, the result is trivial). We first state the reduction to $\lMEE{F/C}{l}$. 

 For this, let $k=\lSIZE l(\psi)$, and let $\phi$ be a $B$-formula. We first test whether there is an assignment that makes at most $k$ variables true and satisfies $\phi$. In this case, the reduction outputs a string that is not a positive instance of $\lMEE{F/C}{l}(B)$. Otherwise, we produce the instance $\pair\phi k$.

 The reduction can be performed in logarithmic space, since $k$ is constant and the truth value of a formula can be determined in logarithmic space~\cite{bus87}. We prove that it is correct: First assume that $\psi$ is unsatisfiable. In that case, the reduction produces the string $\pair\phi k$ which is a positive instance, since $\phi$ is equivalent to $\psi$, and $\lSIZE l(\psi)=k$.

 Now assume that $\phi$ is satisfiable. If there is an assignment that satisfies $\psi$ and has at most $k$ true variables, then the result of the reduction is not a positive instance of $\lMEE{F/C}{l}(B)$ by construction. Hence assume that this is not the case, then the result of the reduction is $\pair\phi k$. Assume indirectly that this is a positive instance. Then $\phi$ is equivalent to a formula or circuit $\chi$ with at most $k$ literals. Since $\phi$ is satisfiable, so is $\chi$. Since at most $k$ literals appear in $\chi$, there is a satisfying assignment of $\chi$ (and thus of $\phi$) that sets at most $k$ variables to true, which is a contradiction.

 The reduction to $\lMEE{F/C}{s}$ is analogous: Let $k=\lSIZE s(\psi)$, and let $n$ be the maximal number of variables in a formula $\chi$ with $\lSIZE s(\chi)\leq k$. Since there are only finitely many formulas with this size, $n$ is a constant. The remainder of the proof is identical to the above case, where instead of $k$ variables, we consider $n$ variables:

 For an input formula $\phi$, we first test whether there is an assignment that makes at most $n$ variables true and satisfies $\phi$. In this case, the reduction outputs a string that is not a positive instance of $\lMEE{F/C}{l}(B)$. Otherwise, we produce the instance $\pair\phi k$. Again, the reduction can be performed in logarithmic space.

 If $\phi$ is unsatisfiable, the reduction produces $\pair\phi n$ which is a positive instance as $\phi\equiv\psi$. Hence assume $\phi$ is satisfiable, and indirectly assume that the reduction produces a positive instance. In particular, $\phi$ cannot be satisfied with at most $n$ variables set to true, and the result of the reduction is $\pair\phi k$. Hence there is a formula $\chi$ with $\lSIZE s(\chi)\leq k$ and $\phi\equiv\chi$. Since $\phi$ is satisfiable, so is $\chi$, and since $\lSIZE s(\chi)\leq k$, we know that at most $n$ variables appear in $\chi$. Hence $\chi$ (and thus $\phi$) has a satisfying assignment with at most $n$ variables set to true, a contradiction.
\end{proof}

Using results on the complexity of $\sat{B}$~\cite{lew79}, we obtain hardness results for a large class of sets $B$:

\begin{corollary}
 Let $B$ be a finite set of Boolean functions such that there is a $B$-formula that is equivalent to $x\wedge\overline y$. Then $\lMEE{F/C}{l/s}(B)$ is $\CONP$-hard.
\end{corollary}

\begin{proof}
 This follows from Theorem~\ref{theorem:unsat red mee lewis} and the result shown in~\cite{lew79}, which proves that $\sat B$ is $\NP$-complete for these choices of $B$.
\end{proof}

\subsection{Hardness Results: Reducing from Equivalence}

The remaining cases are those where satisfiability is tractable, but which are not of the forms covered by Theorem~\ref{theorem:lewis:ptime cases}. We show that in these cases, minimization is $\CONP$-hard using a reduction from the \emph{equivalence problem} for formulas, which asks to determine whether two given formulas are equivalent. We first need a technical lemma that will be used in our constructions later. In the following, a variable $x$ is \emph{relevant} for a function $f$ if the value of $f$ is in fact influenced by the value of the variable, i.e., if there exist assignments $\alpha$ and $\alpha'$ such that $\alpha$ and $\alpha'$ agree on all variables except $x$, and $f(\alpha)\neq f(\alpha')$. Note that the size of the smallest formula is always at least as large as that of the smallest circuit, hence the following result covers the circuit case as well.

\begin{proposition}\label{proposition:minimal size of circuit for non degenerated function}
 Let $B$ be a finite set of Boolean functions such that $B$ contains a function that it at least binary. Let $m$ denote the maximal arity of a function in $B$, and let $l>1$. Then for every $B$-circuit $C$ with $m\cdot l$ relevant input variables, we have that $\lSIZE s(C)\ge l+1$
\end{proposition}

\begin{proof}
 This follows trivially since a connected $B$-circuit with $l$ non-input gates can only connect $m\cdot l-(l-1)<m\cdot l$ input gates. 
\end{proof}

The proof of the theorem below relies on the following idea: Given two formulas as input for the equivalence problem, we combine them into a single formula which is ``trivial'' if the formulas are equivalent, but ``complicated'' otherwise. The ``gap'' between the cases is large enough to yield a reduction to the minimization problem.

\begin{theorem}
 Let $B$ be a finite set of Boolean functions such that $\wedge\in\clone B$ and $\vee\in\clone{B\cup\set 1}$. Then $\lMEE{F/C}{l/s}(B)$ is $\CONP$-hard.
\end{theorem}

\begin{proof}
  From Theorem~4.15 in~\cite{rei01}, we know that the problem of testing whether two given $B$-formulas are equivalent is \CONP-complete. We show that this problem reduces to $\lMEE{F/C}{l/s}(B)$. Since $\vee\in\clone B$ and $\wedge\in\clone{B\cup\set1}$, there are $B$-formulas $f_\vee(x,y,t)$ and $f_\wedge(x,y)$ such that $f_\wedge(x,y)\equiv x\wedge y$, and $f_\vee(x,y,1)\equiv x\vee y$. Let $m$ denote the maximal arity of a function in $B$. We first consider $\lMEE{F/C}s(B)$.

 Let $H_1$ and $H_2$ be $B$-formulas, and define
 \begin{itemize}
  \item $l=\lSIZE s(f_\wedge(H_1,t))$, without loss of generality assume $l>1$.
  \item $Z$ is a $B$-formula equivalent to $\bigwedge_{i=1}^{m\cdot l}z_i$ for new variables $z_i$.
  \item $G=f_\wedge(f_\vee(f_\wedge(H_1,H_2),f_\wedge(f_\vee(H_1,H_2,t),Z),t),t)$.
 \end{itemize}

 Note that $l$ is polynomial in the input, since $m$ is constant, and $f_\wedge(H_1,t)$ clearly can be constructed. Hence the formula $Z$ can be computed in polynomial time as well, since we can represent the conjunction over the $z_i$'s as a tree of logarithmic depth, which grows only polynomially when repeatedly implementing $\wedge$ with $f_\wedge$.

 Also note that by construction, we have $$G\equiv(H_1\wedge H_2)\vee ((H_1\vee H_2)\wedge\bigwedge_{i=1}^{m\cdot l} z_i).$$

 We claim that $H_1\equiv H_2$ if and only if $\pair Gl\in\lMEE{F/C}{s}(B)$.

 First assume that $H_1\equiv H_2$. In this case, $G$ is equivalent to $t\wedge H_1$, which is equivalent to $f_\wedge(H_1,t)$, and thus there is a $B$-formula/circuit equivalent to $G$ with size $l$ by definition of $l$.

 Now assume that $H_1\not\equiv H_2$. Then there is an assignment $\alpha$ that, without loss of generality, satisfies $H_1$ but not $H_2$. In this case, it easily follows that $G[\alpha]$ (i.e., $G$ with the values for $\alpha$ hard-coded into the input gates, which is not necessarily a $B$-circuit anymore) is equivalent to $t\wedge\bigwedge_{i=1}^{m\cdot l} z_i$. Therefore, in this case all of the $z_i$ are relevant variables for $G$. Therefore, Proposition~\ref{proposition:minimal size of circuit for non degenerated function}. implies that for every $B$-formula or circuit $\chi$ equivalent to $G$ we have $\lSIZE{s}(\chi)\ge l+1$.

 The proof for $\lMEE{F/C}{l}(B)$ is identical, except in this case we choose $l$ as the size of literals in $f_\wedge(H_1,t)$, and consider a conjunction of $l$ variables $z_i$. In the positive case, the equivalent formula $f_\wedge(H_1,t)$ has $l$ literals, in the negative case, any formula equivalent to $t\wedge\bigwedge_{i=1}^{l}z_i$ needs to have at least $l+1$ literals.
\end{proof}

We now show an analogous hardness result for the case that $B$ can express the ternary majority function. Algebraically, this condition is equivalent to $\clone B$ containing exactly the Boolean functions $f$ which are self-dual (i.e., $\dual f$ is the same function as $f$) and monotone (i.e., if $\alpha_1\leq\beta_1$, \dots, $\alpha_n\leq\beta_n$, then $f(\alpha_1,\dots,\alpha_n)\leq f(\beta_1,\dots,\beta_n)$).

\begin{theorem}
 Let $B$ be a set of Boolean functions such that $\mathit{maj}\in\clone B$, where $\mathit{maj}(x,y,z)=1$ if and only if $x+y+z\ge 2$. Then $\lMEE{F/C}{l/s}(B)$ is $\CONP$-hard.
\end{theorem}

\begin{proof}
 We show that the equivalence problem for $B$-formulas, which is \CONP-complete due to Theorem~4.17 of~\cite{rei01}, reduces to $\lMEE{F/C}{l/s}(B)$. Again, let $m$ denote the maximal arity of a function in $B$. Since $\mathit{maj}\in\clone B$, there is a $B$-formula $f_{\mathop{maj}}$ such that $f(x,y,z)$ is equivalent to $(x\wedge y)\vee(x\wedge z)\vee(x\wedge z)$. We note that $f_{\mathop{maj}} (x,y,0)\equiv x\wedge y$, and $f_{\mathop{maj}}(x,y,1)\equiv x\vee y$. To increase readability, we also use the symbols $E$ and $V$ for $f_{\mathop{maj}}$ when the last argument is assigned $0$ or $1$, respectively. It follows that $E(x,y,0)\equiv x\wedge y$ and $V(x,y,1)\equiv x\vee y$. We first consider $\lMEE{F/C}{s}(B)$. Hence, let $H_1$ and $H_2$ be $B$-formulas. We construct the following:

 \begin{itemize}
  \item Let $l=\lSIZE{s}(V(f,E(H_1,H_2,f),t))$, where $f$ and $t$ are new variables. Then $l>1$.
  \item let $E^*$ be a formula with variables $z_1,\dots,z_{m\cdot l},f$, such that $E^*(z_1,\dots,z_{m\cdot l},0)=\wedge_{i=1}^{m\cdot l} z_i$,
  \item let $H$ be the formula
    $$V(V(f,E(H_1,H_2,f),t),E(E(t,V(H_1,H_2,t),f),E^*,f),t).$$
 \end{itemize}

 Obviously, $l$ is polynomial in the input and the formula $E^*$ can be computed as follows: Construct the formula $\wedge_{i=1}^{m\cdot l} z_i$ as a $\wedge$-tree of logarithmic depth, and substitute each $\wedge$ with its implementation using $f_{\mathop{maj}}$ and $f$. Then the representation of $E^*$ can be computed in polynomial time.

 We claim that $H_1\equiv H_2$ if and only if $\pair Hl\in\lMEE{F/C}{l/s}(B)$. We consider all possible truth assignments for $f$ and $t$:

 \begin{description}
  \item[If $t=f=0$,]         then $V(f,E(H_1,H_2,f),t)\equiv 0$,                and $E(t,V(H_1,H_2,t),f)\equiv 0$.
  \item[If $t=f=1$,]         then $V(f,E(H_1,H_2,f),t)\equiv 1$,                and $E(t,V(H_1,H_2,t),f)\equiv 1$.
  \item[If $f=0$ and $t=1$,] then $V(f,E(H_1,H_2,f),t)\equiv H_1\wedge H_2$, \\ and $E(t,V(H_1,H_2,t),f)\equiv H_1\vee H_2$.
  \item[If $f=1$ and $t=0$,] then $V(f,E(H_1,H_2,f),t)\equiv H_1\vee H_2$,   \\ and $E(t,V(H_1,H_2,t),f)\equiv H_1\wedge H_2$.
 \end{description}

 In all cases we obtain that if $H_1\equiv H_2$, then $H\equiv V(f,E(H_1,H_2,f),t)$.

 First assume that $H_1\equiv H_2$. In this case, from the definition of $l$ it follows that there is a $B$-formula equivalent to $H$ with size at most $l$.

 Now assume that $H_1\not\equiv H_2$, then there is an assignment $\alpha$ such that, without loss of generality, $\alpha$ satisfies $H_1$ and does not satisfy $H_2$. We extend $\alpha$ with $\alpha(f)=0$ and $\alpha(t)=1$. We then have that $H[\alpha]$ (again, this is $H$ with the values for $\alpha$ hard-coded into it, which is not necessarily a $B$-circuit) is equivalent to $\wedge_{i=1}^{m\cdot l}z_i$, and therefore every $z_i$ is a relevant variable in $H$. Therefore, Proposition~\ref{proposition:minimal size of circuit for non degenerated function} implies that every $B$-circuit equivalent to $H$ has size at most $l+1$.

 The proof for $\lMEE{F/C}{l}(B)$ is identical, except in this case we choose $l$ as the number of literals in $V(f,E(H_1,H_2,f),t)$, and consider a conjunction of $l+1$ variables $z_i$. In the positive case, the equivalent formula $V(f,E(H_1,H_2,f),t)$ has $l$ literals, in the negative case, any formula equivalent to $\bigwedge_{i=1}^{l+1}z_i$ needs to have at least $l+1$ literals.
\end{proof}

\subsection{Classification Theorem}

From the structure of Post's lattice~\cite{pos41} (see~\cite{bcrv03} for a summary), it follows that if $B$ is a finite set of Boolean functions that contains functions $f$, $g$, and $h$ such that $f$ that is not an OR-function, $g$ is not an AND-function, and $h$ is not an XOR-function, then one of the following is true:
\begin{enumerate}
 \item $\wedge\in\clone B$ and $\vee\in\clone{B\cup\set 1}$, or 
 \item $\mathit{maj}\in\clone B$, where $\mathit{maj}$ is the ternary majority function.
\end{enumerate}
 
In both of these cases, the above two theorems imply that the minimization problem is \CONP-hard. Hence, the problem is \CONP-hard for all cases except those covered by our polynomial-time results in Section~\ref{section:lewis:ptime}. We therefore obtain the following full classification:

\begin{corollary}\label{corollary:lewis:complete characterization}
 Let $B$ be a finite set of Boolean functions.
 \begin{itemize}
  \item If $B$ contains only OR-functions, only AND-functions, or only XOR-functions, then $\lMEE{F/C}{l/s}(B)$ can be solved in polynomial time.
  \item Otherwise, $\lMEE{F/C}{l/s}(B)$ is \CONP-hard.
 \end{itemize}
\end{corollary}

\section{Minimization in the CNF framework}

Constraint formulas are CNF-formulas, where the set of allowed types of clauses is defined in a \emph{constraint language} $\Gamma$, which is a finite set of non-empty finitary Boolean relations. A $\Gamma$-clause is of the form $R(x_1,\dots,x_n)$, where $R$ is an $n$-ary relation from $\Gamma$, and $x_1,\dots,x_n$ are variables. A \emph{$\Gamma$-formula} is a conjunction of $\Gamma$-clauses, it is satisfied by an assignment $\alpha$, if for every clause $R(x_1,\dots,x_n)$ in $\varphi$, we have that $(\alpha(x_1),\dots,\alpha(x_n))\in R$. A relation $R$ is \emph{expressed} by a formula if the tuples in the relation are exactly the solutions of the formula (assuming a canonical order on the variables). We denote the satisfiability problem for $\Gamma$-formulas with $\sat{\Gamma}$. We often identify a relation and the formula expressing it. For a constraint language $\Gamma$, we define the constraint language $\overline\Gamma$, which is obtained from $\Gamma$ by exchanging $0$ and $1$ in every relation in $\Gamma$. This language is also called the \emph{dual} of $\Gamma$.

A natural way to measure the size of a CNF formula is the number of clauses---for a fixed language $\Gamma$, this is linearly related to the number of variable occurrences. We thus consider the following problem:

\decisionproblem{$\MEE{\Gamma}$}{A $\Gamma$-formula $\varphi$, an integer $k$}{Is there a $\Gamma$-formula $\psi$ with at most $k$ clauses and $\psi\equiv\varphi$?}

First note that obviously, the duality between $\Gamma$ and $\overline\Gamma$ directly results in these languages leading to the same complexity:

\begin{proposition}\label{prop:duality works}
  Let $\Gamma$ be a constraint language. Then $\MEE\Gamma\eqlogm\MEE{\overline\Gamma}$.
\end{proposition}

To state our classification, we recall relevant properties of Boolean relations (for more background on these properties and how they relate to complexity classifications of constraint-related problems, see e.g.,~\cite{crkhsu01}). 

\begin{enumerate}
  \item A relation is \emph{affine} if it can be expressed by a $\{x,\overline{x},x_1\oplus\dots\oplus x_n,\neg(x_1\oplus\dots\oplus x_n)\ \vert\ n\in\mathbb N\}$-formula.
  \item A relation is \emph{bijunctive} if it can be expressed by a $\Gamma^2$-formula, where $\Gamma^2$ is the set of binary Boolean relations.
  \item A relation is \emph{Horn} if it can be expressed by a $\{x,\overline{x},(x_1\wedge\dots\wedge x_n\rightarrow y),(\overline{x_1\wedge\dots\wedge x_n})\ \vert\ n\in\mathbb N\}$-formula.
  \item A relation is \emph{positive Horn} if it can be expressed by a $\{x_1\wedge\dots\wedge x_n\rightarrow y\ \vert\ n\in\mathbb N\}$-formula.
  \item A relation is \emph{IHSB$+$} if it can be expressed by a $\{x,\overline x,x\rightarrow y,(x_1\vee\dots\vee x_n)\ \vert\ n\in\mathbb N\}$-formula.
\end{enumerate}
A constraint language $\Gamma$ is affine, bijunctive, IHSB$+$, or (positive) Horn if every relation in $\Gamma$ has this property. $\Gamma$ is \emph{dual (positive) Horn} if $\overline\Gamma$ is (positive) Horn, and \emph{IHSB$-$} if $\overline\Gamma$ is IHSB$+$. Note that IHSB$+$ implies dual Horn, and IHSB$-$ implies Horn. Additionally, $\Gamma$ is \emph{Schaefer} if it is affine, bijunctive, Horn, or dual Horn. This property implies tractability of many problems for Boolean constraint languages, including satisfiability \cite{sch78}, equivalence \cite{boherevo02} and enumeration \cite{crhe97}. For the latter two, the Schaefer property is necessary for tractability, unless $\PTIME=\NP$.

\subsection{Irreducible Relations}\label{sect:irreducible relations}

For many problems in the constraint context, it can be shown that if two constraint languages $\Gamma_1$ and $\Gamma_2$ have the same ``expressive power'' (with regard to different notions of expressibility), then the problems for $\Gamma_1$ and $\Gamma_2$ have the same complexity. A lot of work has been done on categorizing relations with regard to their expressive power, which is related to certain algebraic closure properties of the involved relations. For a discussion of the relationship between different notions of expressiveness, see~\cite{SchnoorSchnoor-Galois-LNCS-2008}. One of the strictest notions of ``expressive power'' is the following: We say that constraint languages $\Gamma_1$ and $\Gamma_2$ have the same expressive power if and only if every relation in $\Gamma_1$ can be expressed by a $\Gamma_2$-formula and vice versa. This notion of expressiveness has been studied in \cite{crkoza07:plain-bases}. It arises naturally in many complexity considerations for constraint-related problems: If two constraint languages have the same expressive power, then one can easily show that formulas can be ``translated'' from one language to the other with little computational cost. Hence it is natural that the complexity for all computational problems where the answer remains invariant if input formulas are exchanged for equivalent ones will then be the same---this includes satisfiability, equivalence, enumeration, and many other problems that have been considered. However, the minimization problem behaves differently: When translating formulas between different constraint languages, the number of clauses does not remain invariant. Moreover, even for constraint languages $\Gamma_1$ and $\Gamma_2$ with the same expressive power, expressing the same relation can possibly be done more efficiently using the language $\Gamma_1$ than using $\Gamma_2$. Therefore an easy proof showing that languages with the same expressive power lead to minimization problems with the same complexities cannot be expected. In fact, we show that the statement is not even true, by exhibiting constraint languages which have the same expressive power, yet having different complexities of the minimization problem. However, our complexity classification obtained later still heavily relies on the characterization of Boolean relations along the above lines.

\begin{example}\label{ex:galois does not work}\upshape
Let $\Gamma_1:=\set{x\vee y}$ and $\Gamma_2:=\{(x\vee y),(x\vee (y\wedge z)),(x\vee(y\wedge z\wedge w))\}$. Then obviously, every $\Gamma_2$-formula can be rewritten as a $\Gamma_1$-formula and vice versa, using the equivalence $y\vee(x_1\wedge\dots\wedge x_n)\equiv (y\vee x_1)\wedge\dots\wedge (y\vee x_n)$. However, while the problem $\MEE{\Gamma_1}$ can obviously be solved in polynomial time, the problem $\MEE{\Gamma_2}$ is \NP-hard. This follows from a reduction similar in flavor to the one used in the proof of NP-hardness of MEE in \cite{hewe02}, reducing from the Vertex Cover for cubic graphs problem: A cubic graph $G = (V,E)$ has a vertex cover of size $k$ if and only if the formula $\bigwedge_{\{i,j\} \in E} x_i \vee x_j$ has an equivalent $\Gamma_2$-formula with $k$ clauses.
\end{example}

Therefore, unlike all other problems in the constraint context that we mentioned, the complexity of the minimization problem is not determined by the expressive power of a constraint language. However, the problems that we need to solve in order to minimize $\Gamma_2$-formulas are combinatorial in nature, and do not stem from the difficulty of determining a ``minimum representation'' of what these formulas actually express. Therefore the \NP-hardness is not related to the problem that we are interested in in minimization, namely to find a shortest equivalent formula, but from the difficulty of how to use the ``building blocks'' that we have efficiently. While this certainly is an interesting problem in its own right, in this paper we only study the complexity of the actual task of finding---not expressing---a minimum representation of a formula in the given constraint language.

Looking at the example given above, the problems obviously arise from the fact that $\Gamma_2$ contains ``combined'' relations which can be re-written into simpler clauses: the clause $(x\vee (y\wedge z))$ is equivalent to $(x\vee y)\wedge(x\vee z)$. An important feature in the study of constraint satisfaction problems is that they allow us to build formulas from ``local conditions,'' which are expressed in the individual clauses. The clause $(x\vee (y\wedge z))$ is in a way not ``as local as it can be,'' since it can be rewritten as the conjunction of two ``easier'' conditions. We define \emph{irreducible} relations as those that cannot be rewritten like this:

\begin{definition}
 An $n$-ary relation $R$ is \emph{irreducible}, if for every formula $R_1(x^1_1\dots,x^1_{k_1})\wedge\dots\wedge R_m(x^m_1,\dots,x^m_{k_m})$ (where each $R_i$ is a $k_i$-ary Boolean relation) which is equivalent to $R(x_1,\dots,x_n)$, one of the $R_i$-clauses has arity at least $n$. A constraint language $\Gamma$ is \emph{irreducible} if every relation in $\Gamma$ is.
\end{definition}

The intuition behind the definition is that a relation $R$ is irreducible if the question if some tuple $(\alpha_1,\dots,\alpha_n)$ belongs to $R$ cannot be answered by checking independent conditions which each only depend on a proper subset of the values, but all of the $\alpha_i$ have to be considered simultaneously. Note that the clause which has arity $\ge n$ can be assumed to contain every variable of $x_1,\dots,x_n$, since otherwise, a variable appears twice in the clause, which then could be rewritten with a relation of a smaller arity. Hence a relation is irreducible if and only if every formula equivalent to $R(x_1,\dots,x_n)$ has a clause that contains (at least) all of the variables $x_1,\dots,x_n$.

We mention that it is \emph{not} sufficient to replace, in the above definition of irreducibility, the ``has arity at least $n$'' with ``has arity $n$.'' In this case, no relation would be irreducible at all, since $R(x_1,\dots,x_n)$ can always be expressed with the $n+1$-ary term $R'(x_1,\dots,x_n,x_n)$, where $R'=\set{(x_1,\dots,x_n,y)\ \vert\ (x_1,\dots,x_n)\in R\mathtext{ and }x_n=y}$. We thank the anonymous reviewer of~\cite{Hemaspaandra-Schnoor-Minimization-IJCAI-2011-TOAPPEAR} for pointing out this issue.

Irreducibility is a rather natural condition---in fact, most relations usually considered in the constraint context meet this definition:

\begin{example}
 \begin{enumerate}
  \item Let $R$ be expressed by a disjunction of $n$ literals with $n$ distinct variables. Then $R$ is irreducible.
  \begin{proof}
   Let $R$ be expressed by $(l_1\vee\dots\vee l_n)$, where $l_i$ is either $x_i$ or $\overline{x_i}$ for $n$ distinct variables $x_1,\dots,x_n$. Let $R_1(x^1_1\dots,x^1_{k_1})\wedge\dots\wedge R_m(x^m_1,\dots,x^m_{k_m})$ be a formula equivalent to $R(x_1,\dots,x_n)$. We need to show that there is one clause $R_i(\dots)$ where each of the $x_i$ appears. Let $I$ be the assignment to the variables $x_1,\dots,x_n$ such that $I(x_i)=1$ if $l_i$ is the literal $\overline{x_i}$, and $I(x_i)=0$ if $l_i$ is the literal $x_i$. Then $I$ does not satisfy the clause $R(x_1,\dots,x_n)$. Therefore there must be a clause $R_j(x^j_1,\dots,x^j_{k_j})$ not satisfied by $I$. We show that each variable $x_i$ appears in this clause, which then completes the proof. Let $I'$ be the assignment agreeing with $I$ for all variables except for $x_i$. Then $I'$ satisfies $R(x_1,\dots,x_n)$, and hence satisfies $R_j(x^j_1,\dots,x^j_{k_j})$. Since $I$ and $I'$ only differ in the variable $x_i$, this variable must appear in $R_j(x^j_1,\dots,x^j_{k_j})$.
  \end{proof}
  \item Let $R$ be expressed by a clause $x_1\oplus\dots\oplus x_k=c$ for distinct variables $x_1,\dots,x_k$ and a constant $c\in\set{0,1}$. Then $R$ is irreducible.
  \begin{proof}
   Similar to the above: Let $\varphi$ be a conjunction of clauses equivalent to $R(x_1,\dots,x_k)$. Fix an assignment $I$ not satisfying $\varphi$. Then there must be a clause in $\varphi$ not satisfied by $I$, but changing the truth value of any of the variables makes the formula (and hence this clause) satisfied, therefore every variable appears in the clause.
  \end{proof}
  \item Every relation appearing in the base-list given in \cite{crkoza07:plain-bases} is irreducible.
\end{enumerate}
\end{example}

The above list certainly is not exhaustive. Irreducible languages only allow ``atomic'' clauses that cannot be split up further. In practice, for example in the design of knowledge bases, irreducible languages are more likely to be used: They provide users with atomic constructs as a basis from which more complex expressions can be built. We have seen above that there are languages with equal expressive power (meaning that formulas can be easily rewritten from one of the languages to the other), but the irreducible one has an easier minimization problem than the non-irreducible one. We do not expect that an example exists for the converse, where the problem is easier for the reducible case than for the irreducible, for the reason discussed above: In an informal way, when considering the minimization problem for irreducible languages, it is sufficient to find some minimum representation for the formula. The task to express this formula using the a minimum number of clauses of the given constraint language is easy. As seen in Example~\ref{ex:galois does not work}, in the case of reducible languages, this second task can be \NP-hard in itself.

\subsection{Tractable Cases: Polynomial-Time algorithms}

We now prove polynomial-time results for a wide class of constraint languages. In fact, we prove the maximum of polynomial time results that can be expected: As mentioned before, the MEE problem for positive Horn formulas is \NP-complete~\cite{BC94}. We show the following result: For every irreducible constraint language that is Schaefer, and does not have all the expressive power of positive Horn (or dual positive Horn), the minimization problem can be solved efficiently. This proves polynomial-time results in each case where such a result can be expected, since for non-Schaefer languages, even testing equivalence is \CONP-hard. Following well-known classification results about the structure of Boolean constraint languages, there are three cases to consider (ignoring the isomorphic cases arising to duality, see Proposition~\ref{prop:duality works}): The case where $\Gamma$ is affine, bijunctive, or IHSB$+$. For each of these cases, we prove that the minimization problem can be solved efficiently. The most interesting and involved case is for constraint languages that are IHSB$+$. The bijunctive case is a simpler version of the IHSB$+$-case, the affine case is tractable due to the fact that formulas involving affine constraint languages can be seen as linear equations, for which there are efficient algorithms. 

\subsubsection{IHSB$+$ and IHSB$-$ formulas}

We start our polynomial-time results with the most involved of these constructions, proving that irreducible constraint languages that are IHSB$+$ lead to an easy minimization problem (from Proposition~\ref{prop:duality works}, it follows that the problem is polynomial-time solvable for IHSB$-$ as well). We first prove the result for the basic case where the relations in our constraint language are restricted to the ones ``defining'' IHSB$+$, and later prove that this case already is general enough to cover \emph{all} irreducible languages that are IHSB$+$. Requiring irreducibility is necessary: The language $\Gamma_2$ discussed in Example~\ref{ex:galois does not work} is IHSB$+$ (in fact, considerably less expressive than IHSB$+$), but, as argued before, the minimization problem for $\Gamma_2$ is \NP-hard.

The main idea of the algorithm is the following: We rewrite formulas using multi-ary OR, implication, equality, and literals into conjunctions of, to a large degree, independent formulas, each containing only OR, implications, equalities, or literals. Each of these formulas then can be minimized locally with relatively easy algorithms. The main task that our algorithm performs is ``separating'' the components of the input formula in such a way that minimizing the mentioned sub-formulas locally is equivalent to minimizing the entire formula.

\begin{theorem}\label{theorem:constraint minimization ptime case}
 Let $\Gamma=\set{\rightarrow,=,x,\overline{x}}\cup\set{\textnormal{OR}^m\ \vert\ m\leq k}$ for some $k \in\mathbb N$. Then $\MEE\Gamma\in\PTIME$.
\end{theorem}

\begin{proof}
 We first introduce some notation and facts about $\Gamma$-formulas: For variables $u$ and $v$, we write $u\leadsto_{\varphi} v$ ($u$ \emph{leads to} $v$ \emph{in} $\varphi$) if there is a directed path consisting of $\rightarrow$ and $=$-clauses in the formula $\varphi$ from $u$ to $v$. We often omit the formula and simply write $u\leadsto v$. Similarly, if there are OR-clauses $C_1=(x_1\vee\dots\vee x_n)$ and $C_2=(y_1\vee\dots\vee y_m)$, we write $C_1\leadsto C_2$ if every of the $x_i$ leads to one of the $y_j$. It is obvious that in this case, the conjunction of $C_1$ and the $\rightarrow/=$-clauses implies $C_2$. Note that since $x\leadsto x$ for all variables, it holds that $(x_1\vee x_2)\leadsto(x_1\vee x_2\vee x_3)$. In particular, a literal $x$ leads to a clause $(x\vee y\vee z)$.

 It is easy to see that a $\Gamma$-formula is unsatisfiable if and only if there is some OR-clause (we regard literals as $1$-ary OR-clause) $x_1\vee\dots\vee x_n$ such that for each of the $x_i$, there is a variable $z_i$ which occurs as a negative literal, and $x_i\leadsto z_i$ (otherwise, we can satisfy the formula by setting all variables to $1$ which do not imply negative literals). Since satisfiability for $\Gamma$-formulas can be tested in polynomial time~\cite{sch78}, we assume that all occurring formulas are satisfiable (otherwise in order to minimize we produce a minimum unsatisfiable $\Gamma$-formula, which is a fixed string). For any $\Gamma$-formula $\varphi$, let $\varphi_{\mathtext{OR}}$ denote the formula obtained from $\varphi$ by removing every clause that is not an OR-clause with at least $2$ variables. Similarly let $\varphi_{\rightarrow}$ be the conjunction of all implication-clauses in $\varphi$, $\varphi_{\mathtext{lit}}$ the literals in $\varphi$, and $\varphi_=$ the equality clauses.

 We now describe the minimization procedure. We use some canonical way of ordering variables and clauses (for example, the lexicographical ordering on the names) and repeat the following steps until no changes occur anymore:

\begin{algorithmic}[1]
\STATE{\textbf{Input}: $\Gamma$-formula $\varphi$}
\WHILE{changes still occur}
  \STATE{For a set of variables connected with $=$, only keep the minimal variable in non-equality clauses (by variable identification)}
  \IF{there exist OR-clauses $C_1\neq C_2$ with $C_1\leadsto C_2$,}
    \STATE If $C_2\leadsto C_1$, then remove the minimal of the two
    \STATE Otherwise, remove $C_2$
   \ENDIF
  \IF{there is clause $(x_1\vee\dots\vee x_n)$, variable $v$ with $x_i\leadsto v$ for all $i$,}
    \STATE\label{enumlabel:min-alg:introduce literals} {introduce clause $v$}
    \STATE{remove $\rightarrow$-clauses leading to $v$}
  \ENDIF
  \IF{literal $x$ occurs, $x\leadsto y$}
    \STATE{replace final clause in path with $y$}
  \ENDIF
  \IF{literal $\overline y$ occurs, $x\leadsto y$}
    \STATE\label{enumlabel:min-alg:variables leading to negative literals} replace first clause in path with $\overline{x}$
  \ENDIF
  \STATE\label{enumlabel:min-alg:remove negative literals from or} Remove variables occurring as negative literals from OR-clauses
  \IF{$(x_1\vee\dots\vee x_n)$ is clause, $x_i\leadsto x_j$ for $i\neq j$}
    \STATE{remove $x_i$ from the clause}
  \ENDIF
  \IF {there are variables such that $x_1\leadsto x_2,\dots, x_{n-1}\leadsto x_n,x_n\leadsto x_1$}
    \STATE{exchange implications between them with equalities.}
  \ENDIF
  \IF{$u$ ($\overline u$) appears as a literal}
    \STATE\label{enumlabel:min-alg:remove tautological implications} remove clauses of the form $(v\rightarrow u)$ ($(u\rightarrow v)$).
  \ENDIF
  \STATE Locally minimize $\varphi_=$ and $\varphi_{\rightarrow}$.
\ENDWHILE
\end{algorithmic}

Note that $\varphi_=$ and $\varphi_{\mathtext{lit}}$ can be minimized trivially, and $\varphi_{\rightarrow}$ can be minimized due to a result from~\cite{agu72:trans-reduction}, since finding a transitive reduction of a directed graph is exactly the problem of minimizing a formula in which only implications of positive literals appear.

For a $\Gamma$-formula $\varphi$, let $\min(\varphi)$ denote the result of this optimization procedure on input $\varphi$. It is obvious that $\min(\varphi)$ is equivalent to $\varphi$, and that the algorithm can be performed in polynomial time. It is also obvious that the number of clauses of $\min(\varphi)$ does not exceed the number of clauses of $\varphi$. This is clear from the definition of the algorithm except for step~\ref{enumlabel:min-alg:introduce literals}. In this case, the number of clauses could grow if there is no $\rightarrow$-clause that we can remove. However, in this case, all of the variables in the OR-clause are $=$-connected with $v$, and therefore the clause is equivalent to $v$ and can be removed. Therefore, the number of clauses in $\varphi$ does not increase when applying the algorithm.

The main idea of the algorithm is that it brings the formulas in a ``normal form,'' allowing us to minimize the components of the formula separately. The proof depends on the following claims:

\begin{prooffact}\label{prooffact:implied literals are in formula}
Let $\chi$ a satisfiable $\Gamma$-formula such that $\mathtext{min}(\chi)=\chi$ and $\chi$ implies $x$ ($\overline{x}$) for some variable $x$. Then $\chi_{\mathtext{lit}}$ implies $x$ ($\overline{x}$).
\end{prooffact}

\begin{proof}
First consider the case that $\chi$ implies $x$. Then the formula $\chi\wedge\overline x$ is not satisfiable. Therefore, there is an OR-clause $(x_1\vee\dots\vee x_n)$ such that each $x_i$ leads to a variable $z_i$ which appears as a negative literal. No $x_i$ can lead to a negative literal appearing in $\chi$, since such variables are removed from OR-clauses by the algorithm in steps~\ref{enumlabel:min-alg:variables leading to negative literals} and~\ref{enumlabel:min-alg:remove negative literals from or}. Hence every $x_i$ leads to $x$. Therefore, $x$ is present as a literal due to step~\ref{enumlabel:min-alg:introduce literals} of the minimization algorithm.

Now assume that $\chi$ implies $\overline x$. Then $\chi\wedge x$ is not satisfiable. Hence there must be an OR-clause where every appearing variable leads to a variable occurring as a negative literal. Since $\chi$ is satisfiable, this OR-clause must be the single literal $x$. Hence in $\chi$, $x$ leads to a variable $y$ occurring as a negative literal. By the construction of the minimization algorithm, $\overline x$ then also appears as a literal.
\end{proof}

The following facts are proven using similar arguments:

\begin{prooffact}\label{prooffact:implied implications not in or clauses}
 Let $\chi$ be a satisfiable $\Gamma$-formula such that $\mathtext{min}(\chi)=\chi$, and let $u,v$ be variables in $\Gamma$. such that $\chi$ implies $(u\rightarrow v)$. Then $\chi_{\rightarrow}\wedge\chi_{\mathrm{lit}}\wedge\chi_=$ implies $(u\rightarrow v)$.
\end{prooffact}

\begin{proof}
Assume that this is not the case. It then follows that $\chi\wedge u\wedge\overline v$ is not satisfiable, and $\chi_{\rightarrow}\wedge\chi_{\mathrm{lit}}\wedge\chi_=\wedge u\wedge\overline v$ is. Since the former is unsatisfiable, there is an OR-clause $(x_1\vee\dots\vee x_n)$ such that every $x_i$ leads to a variable occurring as a negative literal in $\chi\wedge u\wedge\overline v$. First assume that this OR-clause is the literal $u$. If $u$ would lead to a variable occurring as a negative literal which is not the variable $v$, then by construction of the algorithm, $\overline u$ would be a literal in $\chi$, a contradiction, since $\chi_{\mathtext{lit}}\wedge u$ is satisfiable. Therefore, we have that $u\leadsto v$, and the claim follows.

Now assume that the OR-clause is not the variable $u$. Since variables leading to negative literals are removed from OR-clauses by the minimization algorithm, all of the $x_i$ lead to $v$. By construction of the algorithm, a literal $v$ is then introduced in $\chi$. This is a contradiction, since $\chi_{\rightarrow}\wedge\chi_{\mathrm{lit}}\wedge\chi_=\wedge u\wedge\overline v$ is satisfiable.
\end{proof}

\begin{prooffact}\label{prooffact:implies equalities in equality clauses}
 Let $\chi$ be a satisfiable $\Gamma$-formula such that $\mathtext{min}(\chi)=\chi$, $\chi$ implies $x=y$, and $\chi$ does not imply $x$ or $\overline x$. Then $\chi_=$ implies $x=y$.
\end{prooffact}

\begin{proof}
From Fact~\ref{prooffact:implied implications not in or clauses}, we know that $\chi_{\mathtext{lit}}\wedge\chi_{\rightarrow}\wedge\chi_=$ implies $(u\rightarrow v)$ and $(v\rightarrow u)$. Therefore, since none of these variables appear as literals, we know that $u\leadsto v$ and $v\leadsto u$. Therefore, $=$-clauses connecting $u$ and $v$ are introduced by the algorithm.
\end{proof}

After establishing these initial facts about the algorithm, we now prove that it is correct, i.e., that $\min(\varphi)$ has a minimal number of clauses among all $\Gamma$-formulas equivalent to $\varphi$. To prove this, let $\psi$ be a formula such that $\psi\equiv\varphi$. We show that $\card{\min(\varphi)}\leq\card\psi$. Since $\card{\min(\psi)}\leq\card\psi$, it suffices to show that $\card{\min(\varphi)}\leq\card{\min(\psi)}$. Hence it suffices to prove that for equivalent formulas $\varphi$ and $\psi$ such that $\min(\varphi)=\varphi$ and $\min(\psi)=\psi$, it follows that $\card{\varphi}\leq\card\psi$.

The main strategy of the remainder of the proof is to show that the above algorithm performs a ``separation'' of the formula in components containing the ``$\rightarrow$-part,'' the ``literal part'' and the ``$=$-part,'' which is, in a sense, uniquely determined: For the two equivalent formulas $\psi$ and $\varphi$, the obtained parts are not necessarily identical, but they are equivalent. This is the main reason why it is sufficient to only minimize these ``components'' in our algorithm.

\begin{prooffact}
 $\mathtext{min}(\psi)_{\rightarrow}\equiv\mathtext{min}(\varphi)_{\rightarrow}$.
\end{prooffact}

\begin{proof}
 Let $(u\rightarrow v)$ be a clause in $\mathtext{min}(\psi)_{\rightarrow}$. Then we know that $\varphi$ does not imply one of $u,\overline{u},v,\overline{v}$: Assume that this is the case. From Fact~\ref{prooffact:implied literals are in formula}, we then know that the corresponding literal appears in the formula itself. In the case that $u$ appears, the clause $(u\rightarrow v)$ would have been deleted by the algorithm, and replaced with the literal $v$. In the case that $\overline{v}$ appears, the clause is replaced with $\overline u$. If $\overline u$ occurs, or $v$ occurs, then the clause $(u\rightarrow v)$ is tautological and has been removed by the algorithm in step~\ref{enumlabel:min-alg:remove tautological implications}. From Fact~\ref{prooffact:implied implications not in or clauses}, we know that (since $\mathtext{min}(\varphi),\mathtext{min}(\psi),\varphi$, and $\psi$ are all equivalent), that $\varphi_{\mathtext{lit}}\wedge\varphi_{\rightarrow}\wedge\varphi_=$ implies $(u\rightarrow v)$. Due to the above, since $u$ and $v$ do not appear as literals and only one variable for each connected $=$-component appears in the rest of the formula, we know that $\varphi_{\rightarrow}$ implies $(u\rightarrow v)$. Therefore, $\varphi_{\rightarrow}$ implies every clause in $\psi_{\rightarrow}$, and hence $\varphi_{\rightarrow}$ implies $\psi_{\rightarrow}$. Due to symmetry, it follows that these formulas are equivalent.
\end{proof}

The following claims follow in a similar way:

\begin{prooffact}
 $\mathtext{min}(\psi)_{\mathtext{lit}}\equiv\mathtext{min}(\varphi)_{\mathtext{lit}}$.
\end{prooffact}

\begin{proof}
 This is obvious from Fact~\ref{prooffact:implied literals are in formula}, since literals appear as literal clauses if and only if they are implied by the formula, and the formulas are equivalent.
\end{proof}

\begin{prooffact}
 $\mathtext{min}(\psi)_=\equiv\mathtext{min}(\varphi)_=$
\end{prooffact}

\begin{proof}
We know from Fact~\ref{prooffact:implied literals are in formula} that a variable $x$ such that $\psi$ implies $x$ or $\overline x$ appears as a positive or negative literal, and a variable appearing as a literal does not appear in $\psi_=$. Let $(u=v)$ be a clause in $\psi_=$. Then $\varphi$ implies $(u\rightarrow v)$ and $(v\rightarrow u)$. Since both do not appear as literals, Fact~\ref{prooffact:implied implications not in or clauses} then implies that $u\leadsto_{\varphi} v$ and $v\leadsto_{\varphi} u$. Therefore, equality clauses between them have been introduced in $\varphi$, and hence $\varphi_=$ implies $(u=v)$. Thus $\varphi_=$ implies $\psi_=$, and due to symmetry, they are equivalent.
\end{proof}

It remains to deal with the OR-components: We want to show that $\varphi_{\mathtext{OR}}$ and $\psi_{\mathtext{OR}}$ are equivalent as well. To show this requires a bit more work. Let $\mathcal C$ be the set of OR-clauses which follow from $\varphi$, and which only contain variables occurring in $\varphi_{\mathtext{OR}}$ (note that we do not have to construct this (potentially exponential) set in the algorithm).

\begin{prooffact}\label{prooffact:minimal clauses appear}
 Let $C$ be a $\leadsto$-minimal clause in $\mathcal C$. Then $C$ appears in $\min(\varphi)$ and in $\min(\psi)$.
\end{prooffact}

\begin{proof}
Let $C=(x_1\vee\dots\vee x_n)$. Since $\varphi$ implies $C$, we know that $\varphi\wedge\overline{x_1}\wedge\dots\wedge\overline{x_n}$ is unsatisfiable. Due to the remarks at the beginning of the proof, this means that there is a clause $B=(y_1\vee\dots\vee y_m)$ such that each of the $y_i$ leads (in $\varphi$) to a variable occurring as a negative literal. Since variables leading to negative literals are removed from OR-clauses by the algorithm, we know that each of the $y_i$ leads to one of the $x_j$. Therefore, $B\leadsto C$. Since $C$ is $\leadsto$-minimal, we know that $C\leadsto B$ holds as well. 

It remains to show that $B$ and $C$ contain the same variables. Assume that there is some variable $x_i$ which does not appear in $B$. Since $B\leadsto C$ and $C\leadsto B$, we know that $x_i$ leads to some variable $y_j$, and that $y_j$ leads to some $x_k$, which in turn leads to some $y_l$. Since $\leadsto$ is transitive, it follows that $y_j\leadsto y_l$. If $y_j$ and $y_l$ would be different variables, then $y_j$ would have been removed from $B$ by the algorithm. Therefore we know that $y_j$ and $y_l$ are the same variables. Since $y_l\leadsto x_k\leadsto y_l$, we know that $\varphi$ implies $(y_l\rightarrow x_k)$ and $(x_k\rightarrow y_l)$, and hence $\varphi$ implies $x_k=y_l$. Since these variables appear in $\varphi_{\mathtext{OR}}$, we know by construction that none of them appears as a literal, and thus, from Fact~\ref{prooffact:implied literals are in formula}, know that neither $x_k,y_l,\overline{x_k}$ or $\overline{y_l}$ are implied by $\varphi$. From Fact~\ref{prooffact:implies equalities in equality clauses}, we therefore know that $\varphi_=$ implies $x_k=y_l$. By construction, only the lexicographically minimal of these two variables appears in $\varphi_{\mathtext{OR}}$, and since both appear, it follows that they are the same variable. This is a contradiction to the assumption that $x_i$ does not appear in $B$. Similarly, we can show that every variable from $B$ appears in $C$.

Since $\psi=\mathtext{min}(\psi)$, and $\varphi$ and $\psi$ are equivalent, the same argument can be used to show that the clause appears in $\psi$.
\end{proof}

We now show the converse of the above fact:

\begin{prooffact}\label{prooffact:all appearing clauses are minimal}
 Let $C$ be a clause appearing in $\mathtext{min}(\varphi)$. Then $C$ is minimal in $\mathcal C$ with respect to $\leadsto$.
\end{prooffact}

\begin{proof}
Assume that this is not the case. Since $\mathcal C$ is a finite set, this implies that there is a minimal clause $B$ in $\mathcal C$ such that $B\leadsto C$. Due to Fact~\ref{prooffact:minimal clauses appear}, we know that $B$ appears in $\varphi_{\mathtext{OR}}$. Since $B\leadsto C$ and $C\not\leadsto B$ (since $C$ is not minimal), $C$ is removed from $\varphi$ by the algorithm, a contradiction.
\end{proof}

Therefore we know that the clauses appearing in $\varphi_{\mathtext{OR}}$ are exactly the minimal clauses in $\mathcal C$, and from Fact~\ref{prooffact:minimal clauses appear}, we know that each of these also appears in a $\psi_{\mathtext{OR}}$. Therefore, the number of OR-clauses in $\varphi_{\mathtext{OR}}$ is bounded by the number of OR-clauses in $\psi_{\mathtext{OR}}$. Due to symmetry, they are equal. Since the components containing literals, equalities and implications have been minimized independently, the number of clauses in $\varphi$ and $\psi$ is equal, which concludes the proof of Theorem~\ref{theorem:constraint minimization ptime case}.
\end{proof}

A careful analysis of the proof yields that it also holds true if $\Gamma$ does not contain all the relations defining IHSB$+$, even though in these cases, only a restricted vocabulary is available for the minimum formula.

\begin{corollary}\label{corollary:ptime for subsets of constraint language}
 Let $\Gamma\subset\set{\rightarrow,=,x,\overline{x}}\cup\set{\textnormal{OR}^m\ \vert\ m\in M}$ for some finite set $M\subseteq\mathbb N$. Then $\MEE\Gamma\in\PTIME$.
\end{corollary}

\begin{proof}
 This follows from the proof of the previous Theorem~\ref{theorem:constraint minimization ptime case}: If a relation is not present in the input formula, it is not introduced (note that positive literals can be written as OR-clauses), except for the case of the equality relation. Simply write this as two implication clauses, and apply the proof of the above theorem, where implications replace equalities. Note that the algorithm reduces OR-clauses in arity, therefore potentially resulting in a clause that cannot be expressed by the constraint language $\Gamma$. However, we can simply use an OR-clause of higher arity with multiple appearances of variables.
\end{proof}

The previous two results covered the case that the constraint language $\Gamma$ contains only the relations that define IHSB$+$. We will now show that irreducible relations that are IHSB$+$ are very close to these ``base relations'' in Lemma~\ref{lemma:irreducible relations in s00 coclone}. This lemma is used in the proof of our main result on IHSB languages, Corollary~\ref{corollary:ihsb}, which shows that our algorithm cannot only be applied to the cases directly covered by Theorem~\ref{theorem:constraint minimization ptime case}, but by every irreducible constraint language that is IHSB$+$ or IHSB$-$. For this, we need some additional notation: We say that a relation $R$ is a \emph{permutation} of a relation $S$ if $R(x_1,\dots,x_n)$ is equivalent to $S(x_{\Pi(1)},\dots,x_{\Pi(n)})$ for some permutation $\Pi$ on the set $\set{1,\dots,n}$.

\begin{lemma}\label{lemma:irreducible relations in s00 coclone}
Let $\Gamma=\set{x,\overline{x},\rightarrow,=,\mathtext{OR}^m\ \vert\ m\in\mathbb N}$. Then every irreducible relation which is IHSB$+$ is a permutation of an element of $\Gamma$.
\end{lemma}

\begin{proof}
 Note that by definition, a relation is IHSB$+$ if and only if it can be expressed by a $\Gamma$-formula (equality can be expressed as two implications). Let $R$ be a relation that can be expressed with a $\Gamma$-formula, and let $n$ be its arity. By choice of $R$, there is a formula $\varphi=R_1(x^1_1\dots,x^1_{k_1})\wedge\dots\wedge R_m(x^m_1,\dots,x^m_{k_m})$ which is equivalent to $R(x_1,\dots,x_n)$, where each $x^i_j$ is an element of $\set{x_1,\dots,x_n}$, and $R_i\in\Gamma$. Without loss of generality, we assume that no clause in $\varphi$ can be removed without changing the represented relation, and that no variable appears twice in an OR-clause, and that no variable can be removed from an OR-clause without changing the represented relation.

 Since $R$ is irreducible, there is a clause $C$ in which every variable appears. First assume that this clause is an $\mathtext{OR}^m$-clause, hence $C=(x_1\vee\dots\vee x_n)$. If no other clause appears in $\varphi$, then $R$ is the $n$-ary OR-relation, and hence an element of $\Gamma$, as required. Therefore assume that there is another clause $C'$ in $\varphi$. Due to the minimality of $\varphi$, $C'$ is not equivalent to $C$. If $C'$ is an OR-clause, then $C'$ contains a proper subset of the variables occurring in $C$ (since in $C$, all variables occur), and hence the clause $C$ is redundant, a contradiction to the minimality of $\varphi$ (note that this also covers the case where $C'$ is a positive literal). If $C'$ is a negative literal $\overline{x_i}$, then $x_i$ can be removed from the clause $C$, a contradiction to the minimality. Therefore assume that $C'$ is an implication, $C'=(x_i\rightarrow x_j)$. Since there are no superfluous clauses in $\varphi$, we know that $x_i$ and $x_j$ are different variables. Then the variable $x_i$ can be removed from the OR-clause $C$ without changing the represented relation, a contradiction.

Now assume that $C$ is not an OR-clause, hence $C$ is a literal or an implication. In particular, the arity of $R$ is at most $2$. If $R$ is a $1$-ary relation, then $R$ obviously is irreducible. Hence assume that $R$ is one of the $16$ binary Boolean relations. We make a complete case distinction. The empty relation cannot be an element of a constraint language by definition. If $R$ only contains a single element, it can be written as a conjunction of literals and therefore is not irreducible. If $R$ is the full binary relation over the Boolean domain, it can be written as $\top(x_1)\wedge\top(x_2)$, where $\top$ is the $1$-ary relation $\set{(0),(1)}$, and hence is not irreducible. It remains to consider the cases where $R$ has exactly $2$ or exactly $3$ elements.

The relation $\set{(0,0),(0,1)}$ is not irreducible, since it can be written as $\overline{x_1}\wedge\top(x_2)$. Similarly, $\set{(0,0),(1,0)}$ is represented by $\top(x_1)\wedge\overline{x_2}$. The relation $\set{(0,0),(1,1)}$ is the equality relation and an element of $\Gamma$, the relation $\set{(0,1),(1,0)}$ is not IHSB$+$, $\set{(0,1),(1,1)}$ is not irreducible (it can be written as $\top(x_1)\wedge x_2$), similarly $\set{(1,0),(1,1)}$ can be written as $x_1\wedge\top(x_2)$.

Now consider the relations with exactly three elements: $\set{(0,0),(0,1),(1,0)}$ is the binary NAND and therefore not IHSB$+$, $\set{(0,0),(0,1),(1,1)}$ is the implication and therefore an element of $\Gamma$, $\set{(0,0),(1,0),(1,1)}$ is a permutation of the implication, and $\set{(1,0),(0,1),(1,1)}$ is the binary OR and hence an element of $\Gamma$.
\end{proof}

The previous two theorems and Proposition~\ref{prop:duality works} directly imply the following corollary, which as mentioned is our main result for IHSB$+$/IHSB$-$ constraint languages:

\begin{corollary}\label{corollary:ihsb}
 Let $\Gamma$ be an irreducible constraint language which is IHSB$+$ or IHSB$-$. Then $\MEE\Gamma\in\PTIME$.
\end{corollary}

\subsubsection{Bijunctive Formulas}

We now cover the final of our polynomial-time cases, which covers constraint languages which are bijunctive. Note that this is not the same as only showing that general 2CNF formulas have an efficient minimization procedure (which was shown in \cite{chang}): In addition to being able to minimize arbitrary 2CNF, we also need to be careful about only using those relations that are present in the constraint language. Again, Example~\ref{ex:galois does not work} shows that the prerequisite that $\Gamma$ is irreducible is necessary.

\begin{theorem}
  Let $\Gamma$ be a constraint language which is irreducible and bijunctive. Then $\MEE\Gamma\in\PTIME$.
\end{theorem}

\begin{proof}
  Since $\Gamma$ is bijunctive, every relation in $\Gamma$ can be written as a formula using only at most binary relations. Since $\Gamma$ is also irreducible, this implies that every relation in $\Gamma$ is at most binary. The only irreducible binary and unary relations over the Boolean domain are (up to permutation of the variables) the literals $x$ and $\overline x$, the binary OR, binary NAND, implication, equality, and exclusive OR. Since all of these relations can be written as implications between literals, minimization can be performed analogously to the proof of the previous Theorem~\ref{theorem:constraint minimization ptime case}. 
\end{proof}

\subsubsection{Affine Formulas}

We conclude our polynomial-time results with the affine case. Affine formulas represent linear equations over $\GF 2$. We therefore can apply linear algebra techniques to obtain an efficient minimization algorithm. Results for linear equations have been obtained before~\cite[Section 8]{curtis84:linearalgebrabook}. We show here that the result covers all cases where the language is affine and irreducible.

\begin{theorem}\label{theorem:affine minimization ptime result}
 Let $\Gamma$ be an irreducible and affine constraint language. Then $\MEE\Gamma\in\PTIME$.
\end{theorem}

\begin{proof}
 Let $\varphi$ be a $\Gamma$-formula. Since satisfiability testing for affine formulas can be done in polynomial time, we can without loss of generality assume that $\varphi$ is satisfiable. Since equivalence for affine formulas can be checked in polynomial time, we can compute a formula which is equivalent to $\varphi$ and irredundant in the sense that if we remove a clause, it is not equivalent to $\varphi$ anymore. Therefore, it suffices to prove that such an irredundant formula already is minimum. Note that a minimum formula obviously is irredundant. We therefore show that two affine formulas $\varphi_1$ and $\varphi_2$ with $\card{\var{\varphi_1}}=\card{\var{\varphi_2}}$ which are both irredundant and are equivalent, have the same number of clauses. In order to do this, we show that a satisfiable, irredundant formula $\varphi$ over $n$ variables with $k$ clauses has exactly $2^{n-k}$ solutions. Let the clauses be $C_1,\dots,C_k$. Since every relation in $\Gamma$ is irreducible, each clause is of the form $x^C_1\oplus\dots\oplus x^C_l\oplus c^C$ for variables $x^C_1,\dots,x^C_l$ and a constant $c\in\set{0,1}$. This can equivalently be written as $x^C_1=\neg(x^C_2\oplus\dots\oplus x^C_l\oplus c^C)$. Since $\varphi$ is irredundant, we know that no clause $C_i$ follows from the clauses $C_1,\dots,C_{i-1}$. Therefore, each clause restricts the possibilities of the values of $x^C_1$, and therefore the relation $R$ represented by $C_1\wedge\dots\wedge C_i$ is a proper subset of the relation $R'$ represented by $C_1\wedge\dots\wedge C_{i-1}$. Since these relations are represented by affine formulas, their cardinalities are powers of $2$. Therefore, $\card{R}\leq\frac{\card{R'}}2$. Since only one variable is restricted in the clause $C_i$, it follows that $\card{R}=\frac{\card{R'}}2$, as claimed.
\end{proof}

\subsection{Hardness Results}

As mentioned before, our polynomial-time results cover all cases where polynomial-time algorithms can be expected. We now prove hardness results for most of the remaining cases.

\subsubsection{Minimization and Satisfiability}

For unrestricted propositional formulas, the MEE problem is obviously \CONP-hard, since a formula $\varphi$ is unsatisfiable if and only if the all-$1$-assignment does not satisfy it, and it has a minimum equivalent expression of size $0$ (which then only can be the constant $0$). The following result uses the same idea of reducing the complement of the satisfiability problem to the minimization problem---however, since the constant $0$ is usually not available in our constraint languages and we are minimizing the number of clauses, the proof is a bit more involved, while still following the same pattern.

\begin{theorem}
 Let $\Gamma$ be a finite constraint language. Then $\overline{\sat\Gamma}\redpm\MEE\Gamma$.
\end{theorem}

\begin{proof}
 Let $\varphi_{\mathtext{min}}$ be an unsatisfiable $\Gamma$-formula with a minimal number of clauses. Let $k_{\mathtext{min}}$ be the number of clauses in $\varphi_{\mathtext{min}}$.

The reduction works as follows: Let $\varphi$ be a $\Gamma$-formula. First compute the set $M$ containing of all $\Gamma$-formulas containing at most $k_{\mathtext{min}}$ clauses with variables appearing in $\varphi$. Note that, since $\Gamma$ is a finite constraint language, this is a polynomial set, and each formula in $M$ has a number of appearing variables bounded by a constant. Therefore we can, in polynomial time, construct for each $\psi\in M$ the set of all solutions $I_\psi$ of $\psi$. For such a solution, let $I^{\mathtext{ext}}_\psi$ be the assignment which agrees with $I_\psi$ for all variables appearing in $\psi$, and assigns $0$ to all other variables.

For each assignment $I^{\mathtext{ext}}_\psi$, check if it is a solution of $\varphi$. If this is the case, let $(\varphi',k')$ be a negative instance of $\MEE\Gamma$. Otherwise, let $(\varphi',k'):=(\varphi,k_{\mathtext{min}})$. We show that $\varphi$ is unsatisfiable if and only if $(\varphi',k')\in\MEE\Gamma$.

First assume that $\varphi$ is unsatisfiable. In particular, in this case it holds that $(\varphi',k')=(\varphi,k_{\mathtext{min}})$. We can, without loss of generality, assume that $\varphi_{\mathtext{min}}$ contains at most one variable, since formulas obtained from unsatisfiable formulas via variable identification remain unsatisfiable. In particular, we can assume that in $\varphi_{\mathtext{min}}$, only variables from $\varphi$ appear. Since $\varphi$ is unsatisfiable, $\varphi$ is equivalent to $\varphi_{\mathtext{min}}$, and hence $(\varphi',k')\in\MEE\Gamma$.

Now assume that $\varphi$ is satisfiable, and assume indirectly that $(\varphi',k')\in\MEE\Gamma$. By choice of $(\varphi',k')$, this implies that $(\varphi',k')=(\varphi,k_{\mathtext{min}})$. Since $M$ contains all $\Gamma$-formulas with at most $k_{\mathtext{min}}$ clauses, it follows that $\varphi$ is equivalent to some formula $\psi\in\Gamma$. Since $\varphi$ is satisfiable, so is $\psi$. Therefore there is some $I_{\psi}$ such that $I_{\psi}$ satisfies $\psi$. Since $\psi$ and $\varphi$ are equivalent, it follows that $I^{\mathtext{ext}}_{\psi}$ satisfies $\varphi$. This is a contradiction, since in this case, the reduction does not produce the instance $(\varphi,k_{\mathtext{min}})$.
\end{proof}

If a constraint language $\Gamma$ is not Schaefer (i.e., neither Horn, dual Horn, bijunctive, nor affine), then the satisfiability problem for $\Gamma^+=\Gamma\cup\set{x,\overline x}$ ($\Gamma$ extended with the possibility to express literals) is \NP-complete. The previous theorem therefore yields the following corollary:

\begin{corollary}
  Let $\Gamma$ be a constraint language that is not Schaefer. Then $\MEE{\Gamma^+}$ is \CONP-hard.
\end{corollary}

\subsubsection{\NP-completeness Results}

In this section we consider the MEE problem for irreducible constraint languages that are Horn, but not IHSB$-$. We show that for these languages, the MEE problem is \NP-complete. This shows that the algorithm we developed in the previous section for the IHSB$+$/IHSB$-$ -case cannot be modified to work with larger classes of formulas (remember that IHSB$-$ formulas are a subset of Horn formulas). Due to Proposition~\ref{prop:duality works}, the analogous result is true for dual Horn and IHSB$+$. We first prove a result about what the irreducible relations here look like:

\begin{theorem}\label{theorem:horn languages contain real horn clause}
 Let $\Gamma$ be an irreducible constraint language such that $\Gamma$ is Horn, but not IHSB$-$. Then there is a relation $R\in\Gamma$ which can be expressed by $x_1\wedge\dots\wedge x_k\rightarrow y$, for $k\ge 2$.
\end{theorem}

\begin{proof}
Let $\Gamma_{\mathtext{Horn}}:=\set{\mathtext{NAND}^k,(x_1\wedge\dots\wedge x_k\rightarrow y)\ \vert\ k\in\mathbb N}$. Since $\Gamma$ is Horn, it follows from~\cite{crkoza07:plain-bases} that every relation in $\Gamma$ can be written as a $\Gamma_{\mathtext{Horn}}$-formula. We show that every relation in $\Gamma$ with an arity of $n\ge 3$ is an element of $\Gamma_{\mathtext{Horn}}$. Therefore, let $R$ be such a relation, and let $\varphi$ be a $\Gamma_{\mathtext{Horn}}$-formula representing $R$, i.e., a formula equivalent to $R(x_1,\dots,x_n)$, and assume that $\varphi$ is minimal in the sense that no clause can be deleted, and no variable can be removed from a clause without changing the relation expressed by the formula.

Since $R$ is irreducible, there is a clause $C$ in $\varphi$ such that every variable $x_1,\dots,x_n$ appears in $C$. If $C$ is the only clause in $\varphi$, then it follows that $R\in\Gamma_{\mathtext{Horn}}$ as claimed. Hence assume that there is another clause $C'$ in $\varphi$. The variables of $C'$ then must be a subset of the variables in $C$. We make a case distinction.

First assume that both clauses are NAND-clauses. If $C$ and $C'$ contain the same variables, then $C$ and $C'$ are equivalent, a contradiction. Therefore the variables appearing in $C'$ are a proper subset of the variables from $C$. Hence $C'$ implies $C$, and $C$ can be removed from $\varphi$ without changing the represented relation, a contradiction.

Now assume that $C$ is a NAND-clause and $C'$ is of the form $(x_{i_1}\wedge\dots\wedge x_{i_k}\rightarrow x_{i_j})$. Then the variable $x_{i_j}$ can be removed from the clause $C$, a contradiction to the minimality of $\varphi$.

Assume that $C$ is (without loss of generality) of the form $(x_1\wedge\dots\wedge x_{n-1}\rightarrow x_n)$, and $C'$ of the form $(x_{i_1}\wedge\dots\wedge\dots x_{i_k}\rightarrow x_j)$. If $x_j$ and $x_n$ are the same variable, then $C'$ implies $C$, and $C$ can be removed from $\varphi$ without changing the relation expressed by the formula, a contradiction to the minimality of $\varphi$. Hence assume that $x_j$ is one of the variables $x_1,\dots,x_{n-1}$. Then the clause $(x_1\wedge\dots\wedge x_{n-1}\rightarrow x_n)$ can be replaced with $(x_1\wedge\dots\wedge x_{j-1}\wedge x_{j+1}\wedge\dots\wedge x_{n-1}\rightarrow x_n)$, a contradiction to the minimality of $\varphi$.

Finally assume that $C$ is of the form $(x_1\wedge\dots\wedge x_{n-1}\rightarrow x_n)$, and $C'$ is a NAND-clause, let $C'=\mathtext{NAND}(x_{i_1},\dots,x_{i_k})$. First assume that the variables in $C'$ contain the variable $x_n$, without loss of generality assume that $i_1=n$. We prove that $C$ can be replaced by the clause $C''=(\overline{x_1}\vee\dots\vee\overline{x_{n-1}})$, which is a contradiction to the minimality of $\varphi$. Therefore, let $I$ be an assignment satisfying $C$ and $C'$, and indirectly assume that $I\nmodels C''$. Then $I(x_1)=\dots=I(x_{n-1})=1$. Since $I\models C$, it follows that $I(x_n)=1$, and hence $I$ does not satisfy $C'$, a contradiction. For the other direction, it is obvious that $C''$ implies $C$. Therefore it remains to consider the case that the variables in $C'$ do not contain the variable $x_n$. In this case it is obvious that $C'$ implies $C$, and hence $C$ can be removed from $\varphi$, a contradiction.

We therefore have proven that every element of $\Gamma$ of arity at least $3$ is an element of $\Gamma_{\mathtext{Horn}}$. In order to prove the theorem, assume that $\Gamma$ does not contain a relation of the required form. Then every relation in $\Gamma$ is either at most binary, or a NAND-relation of some arity. The only non-empty binary relations (up to permutation) over the Boolean domain which are Horn and irreducible are the following:

\begin{itemize}
 \item Conjunctions of literals,
 \item the full relation,
 \item implication, 
 \item $\overline{x}\wedge\top(y),x\wedge\top(y)$,
 \item equality,
 \item binary NAND.
\end{itemize}

(the exclusive-OR relation and the binary OR are not invariant under conjunction, and therefore not Horn). Therefore, $\Gamma$ can only contain NANDs and the relations in the list above, this implies that $\Gamma$ is IHSB$-$, a contradiction.
\end{proof}

The previous theorem shows that every constraint language that is Horn but not IHSB$-$ contains a clause which follows the same pattern as the clauses defining positive Horn. Hence it is not surprising that the proof of the main result of~\cite{BC94} can also be used to show the following:

\begin{theorem}\label{theorem:horn minimization is np hard for cnf}
 Let $\Gamma$ be an irreducible constraint language such that $\Gamma$ is Horn, but not IHSB$-$. Then $\MEE\Gamma$ is \NP-complete.
\end{theorem}

\begin{proof}
The problem is in \NP, since equivalence testing for Horn formulas can be performed in polynomial time~\cite{boherevo02}. \NP-hardness follows using techniques from~\cite{BC94}: One of the main results of that paper is that there exists a reduction $f$ from the well-known \NP-complete Hamiltonian path problem to a DNF minimization problem which has the following properties. We say that a formula is a \emph{pure-Horn-$3$-DNF} if is is a disjunction of clauses, and each clause is a conjunction of $2$ or $3$ literals, with exactly one negative literal.

\begin{enumerate}
 \item For each graph $G$ with $m$ edges, the formula $f(G)$ is a pure-Horn-$3$-DNF
 \item If $G$ has a Hamiltonian path, then there is a pure-Horn-$3$-DNF containing $m+2$ clauses equivalent to $f(G)$,
 \item If $G$ does not have a Hamiltonian path, then there is no DNF containing at most $m+2$ clauses equivalent to $f(G)$.
\end{enumerate}

We describe the obvious procedure to use this result as a proof of the hardness result for $\MEE\Gamma$. It is obvious that the negation of a pure-Horn-$3$-DNF formula $\varphi$ can be written as a $\Gamma$-formula $\mathtext{CNF}(\varphi)$, as the conjunction of the following clauses:

\begin{itemize}
 \item For a clause $x\wedge\overline{y}$ in $\varphi$, introduce a clause $(x\rightarrow y)$,
 \item For a clause $x\wedge y\wedge\overline{z}$ in $\varphi$, introduce a clause $(x\wedge y\rightarrow z)$.
\end{itemize}

These clauses can obviously be constructed using the relation $(x_1\wedge\dots\wedge x_k\rightarrow y)$ and variable identification. We claim that $G$ has a Hamiltonian path if and only if $(\mathtext{CNF}(f(G)),m+2)$ is a positive instance of $\MEE\Gamma$.

First assume that $G$ has a Hamiltonian path. Then, due to the above, $f(G)$ has an equivalent pure-Horn-$3$-DNF formula $\psi$ with at most $m+2$ clauses. Since $\psi$ is equivalent to $f(G)$, it follows that $\mathtext{CNF}(\psi)$ is equivalent to $\mathtext{CNF}(f(G))$, and also hat at most $m+2$-clauses. For the other direction, assume that there is a $\Gamma$-formula $\psi$ with at most $m+2$ clauses which is equivalent to $\mathtext{CNF}(f(G))$. From the case distinction in the proof of Theorem~\ref{theorem:horn languages contain real horn clause}, it is easy to see that every relation in $\Gamma$ can be written as a disjunction of literals. Therefore, the negation of $\psi$ is a DNF-formula with at most $m+2$ clauses which is equivalent to $f(G)$. From the above, it follows that $G$ has a Hamiltonian path, which completes the proof of the theorem.
\end{proof}

The NP-hardness result for constraint languages dealing with Horn logics now follows as a corollary:

\begin{corollary}
 Let $\Gamma$ be an irreducible constraint language that is Schaefer, not affine, not bijunctive, not IHSB$+$, and not IHSB$-$. Then $\MEE{\Gamma}$ is \NP-complete.
\end{corollary}

\begin{proof}
From the well-known classification of constraint languages with respect to their expressive power, it follows that $\Gamma$ is either Horn and not IHSB$-$, or dual Horn and not IHSB$+$. For the first case, the result follows from the theorems in this section, Proposition~\ref{prop:duality works} then implies the result for the dual Horn case.
\end{proof}

\subsection{Classification Theorem}

We can now state our main classification theorem---it follows from the results in the previous sections, and the fact that, by definition, a constraint language which is not affine, bijunctive, IHSB$+$, IHSB$-$, Horn, or dual Horn, is not Schaefer.

\begin{theorem}
  Let $\Gamma$ be an irreducible constraint language.
  \begin{enumerate}
    \item If $\Gamma$ is affine, bijunctive, IHSB$+$, or IHSB$-$, then $\MEE\Gamma\in\PTIME$.
    \item Otherwise, if $\Gamma$ is Horn or dual Horn, then $\MEE{\Gamma}$ is \NP-complete,
    \item Otherwise, $\Gamma$ is not Schaefer, and $\MEE{\Gamma^+}$ is \CONP-hard.
  \end{enumerate}
\end{theorem}

While the theorem does not completely classify the complexity of the MEE problem for all irreducible constraint languages, we consider it unlikely that there exist more polynomial-time cases than the ones we discovered: 
To the best of our knowledge, no decision problem for non-Schaefer constraint languages has been proven to be in polynomial time except for trivial cases (satisfiability of $\Gamma$-formulas can be tested in polynomial time if every relation from $\Gamma$ contains the all-$0$ or all-$1$-tuple). Also, for these languages $\Gamma$, already testing equivalence of formulas is \CONP-hard. This implies that, unless $\PTIME=\NP$, there cannot be a polynomial-time algorithm that, given a $\Gamma$-formula, computes its ``canonical'' (i.e., up to differences checkable by a polynomial-time algorithm) minimum equivalent expression (this would immediately solve the equivalence problem in polynomial time). We are therefore confident that our classification covers all polynomial-time cases for irreducible constraint languages.

It is worth noting that the prerequisite that $\Gamma$ is irreducible is certainly required for the polynomial-time cases, as the earlier example highlighted. For the hardness results, this is less clear---the \CONP-hardness does not rely on this prerequisite at all, and for the \NP-complete Horn cases, we consider it unlikely that there is a constraint language with the same expressive power that does not directly encode positive Horn.

\section{Conclusion and Open Questions}

We have studied the complexity of the minimization problem for restricted classes of propositional formulas in two settings, obtained a complete characterization of all tractable cases in the Post case, and a large class of tractable cases in the constraint case. 

Open questions include the exact classification of the \CONP-hard cases. It is likely that most of them are \NP-hard as well. It would be very interesting to determine whether some of these are actually $\SigPH 2$-complete (this does not follow directly from the $\SigPH 2$-completeness of the minimization problem for CNF formulas~\cite{uma01}, since our constraint languages $\Gamma$ and bases $B$ are finite).

Finally, it would be very interesting to understand how non-irreducibility influences the complexity.

\section*{Acknowledgment}

We thank the anonymous reviewers of~\cite{Hemaspaandra-Schnoor-Minimization-IJCAI-2011-TOAPPEAR} for many helpful comments, in particular for pointing out an issue with the definition of irreducibility.

\bibliographystyle{alpha}
\bibliography{minimization}

\end{document}